\documentclass[11pt]{article}

\usepackage[T1]{fontenc}
\usepackage{lmodern}
\usepackage{amsmath,amssymb,amsthm,mathtools}
\usepackage[round,authoryear]{natbib}
\usepackage{microtype}
\usepackage[a4paper,margin=1.15in]{geometry}
\usepackage[
    colorlinks=true,
    citecolor=blue,
    linkcolor=blue,
    urlcolor=blue
]{hyperref}

\newtheorem{theorem}{Theorem}[section]
\newtheorem{lemma}[theorem]{Lemma}
\newtheorem{proposition}[theorem]{Proposition}
\newtheorem{corollary}[theorem]{Corollary}

\theoremstyle{definition}
\newtheorem{definition}[theorem]{Definition}
\newtheorem{example}[theorem]{Example}

\theoremstyle{remark}

\newcommand{\M}{\mathcal{M}}
\newcommand{\R}{\mathbb{R}}
\newcommand{\emptybundle}{\varnothing}

\title{Robust Welfare Decentralization under Population Entry}
\author{Yi-You Yang\thanks{Department of Applied Mathematics, Chung Yuan Christian University, Taoyuan City, Taiwan. E-mail address: yyyang@cycu.edu.tw}}

\begin{document}

\maketitle

\begin{abstract}
We ask when an incumbent economy with indivisible goods can accommodate an
arbitrary new participant while retaining an efficient allocation supported
by anonymous item prices. We call this property universal entry robustness.
An economy is universally entry-robust if and only if its aggregate welfare
valuation is additive. Although the requirement quantifies over all entrant
valuations, it can be tested using one canonical entrant whose value for a
bundle equals the loss in maximal incumbent welfare caused by removing that
bundle. When the characterization holds, one uniquely determined price vector
decentralizes the incumbent optimum and every one-agent extension. The proof
is direct and uses only demand optimality and welfare comparisons, making
transparent why arbitrary-entry robustness eliminates all bundle interactions
in aggregate welfare. We finally relate this argument to the robust-integrality
interpretation obtained from the linear-programming characterization of
Bikhchandani and Mamer (1997). Individual incumbent valuations may be
nonadditive and need not satisfy gross substitutes.
\end{abstract}


\medskip


\section{Introduction}

Markets are rarely closed populations. New participants arrive with
preferences that need not resemble those of incumbent agents. In markets with
indivisible goods, such entry raises a basic welfare question: can the economy
continue to implement an efficient allocation through anonymous item prices,
regardless of the entrant's valuation?

Competitive-equilibrium existence with indivisible goods has been studied
from several complementary perspectives. \citet{BikhchandaniMamer1997}
provide a general linear-programming characterization for a fixed economy:
market-clearing prices exist exactly when the relevant welfare problem has an
integral optimum in its linear relaxation. The gross-substitutes condition
provides a primitive domain restriction under which equilibrium always exists
\citep{KelsoCrawford1982,GulStacchetti1999}. More recent work develops
broader domain-wide conditions through demand types, unimodularity, and the
structure of substitutability and complementarity across bundles
\citep{BaldwinKlemperer2019,JagadeesanTeytelboym2025}.

Our question changes the quantifier structure of the existence problem. We
fix an incumbent economy and vary only the valuation of a single entrant. The
incumbent economy is called \emph{universally entry-robust} if it admits a
competitive equilibrium and continues to admit one after every possible
one-agent extension. The entrant may have any normalized monotone valuation
and may therefore exhibit unrestricted complementarities. This formulation
produces a particularly transparent answer and permits a proof that does not
invoke linear-programming duality or complementary slackness.

Our characterization is stated in terms of the incumbent economy's aggregate
welfare valuation. For every set of goods $B$, let $v_N(B)$ be the greatest
sum of incumbent values attainable by allocating the goods in $B$ among the
incumbents. Hence, $v_N(B)$ is the maximal utilitarian welfare that the
incumbent population can generate from the resource set $B$. We prove
\[
\boxed{
\text{universal entry robustness}
\iff
\text{additivity of }v_N.
}
\]
Additivity means that every good has a social marginal value that is
independent of which other goods remain available and of how those goods are
reallocated among the incumbents. The relevant restriction is therefore
social rather than individual. Each incumbent may have a nonadditive
valuation, yet the economy can be robust if its aggregate welfare
opportunities are linear in the available goods.

The necessity result is established through a canonical stress test. Define
an entrant valuation by
\[
v_N^{\dagger}(B)
=
v_N(A)-v_N(A\setminus B).
\]
This entrant values a bundle by the loss in maximal incumbent welfare caused
by removing that bundle from the incumbent population. We show that the
incumbent economy is universally entry-robust if and only if it can
accommodate this single entrant. In the associated representative two-agent
economy, every division of the goods generates the same total welfare. If
anonymous item prices support one efficient division, they must support all
efficient divisions and therefore reproduce $v_N$ on every bundle. This is
possible exactly when $v_N$ is additive.

The sufficiency result has a direct decentralization interpretation. Suppose
that
\[
v_N(B)=q(B)=\sum_{a\in B}q_a
\qquad
\text{for every }B\subseteq A.
\]
At prices $q$, an arbitrary entrant chooses a demanded bundle $Y$. The
remaining goods $A\setminus Y$ can then be allocated efficiently among the
incumbents so that each incumbent demands her assignment at the same prices.
Moreover,
\[
w(Y)+v_N(A\setminus Y)
=
q(A)+w(Y)-q(Y),
\]
so the entrant's demand choice also selects a welfare-maximizing division
between the entrant and the incumbent population. Entry may change the
allocation, but the welfare-supporting price system need not change.

Our contribution has five parts. First, we formulate a population-robustness
question for a realized economy: holding the incumbent economy fixed, when
does competitive equilibrium survive every possible single-agent entry?
Second, we obtain a primitive aggregate-welfare characterization: universal
entry robustness holds if and only if the incumbents' aggregate welfare
valuation is additive. Third, although the robustness requirement quantifies
over all entrant valuations, it can be tested using one endogenous
welfare-loss entrant. Fourth, when the characterization holds, one common
price vector supports the incumbent economy and every one-agent extension.
Finally, we provide a direct welfare proof, based only on efficiency and
individual demand inequalities, which makes transparent why arbitrary-entry
robustness eliminates every bundle interaction in aggregate welfare.

The characterization can also be recovered from the linear-programming
framework of \citet{BikhchandaniMamer1997}. We discuss this alternative route
after presenting the direct proof. The two approaches are complementary:
their result supplies a general existence criterion for arbitrary fixed
economies, whereas our argument isolates and explains the economic mechanism
behind a particular robustness property. The paper is also related to work on
robust predictions of competitive equilibrium under indivisibilities. For
example, \citet{JagadeesanEtAl2020} derive equilibrium-independent
participation predictions in exchange economies with transferable utility.
Our robustness notion is instead cross-population: it asks whether equilibrium
existence and efficient price decentralization persist when the participant
set changes.

The remainder of the paper is organized as follows. Section~\ref{sec:model}
introduces the model and universal entry robustness.
Section~\ref{sec:preliminaries} gives two preliminary lemmas.
Section~\ref{sec:canonical} constructs the canonical welfare-loss entrant.
Section~\ref{sec:characterization} proves the characterization theorem.
Section~\ref{sec:welfare} develops its welfare interpretation.
Section~\ref{sec:existing} relates the result to existing equilibrium-
existence characterizations, and Section~\ref{sec:conclusion} concludes.

\section{Model and population entry}
\label{sec:model}

Let $N$ be a finite nonempty set of incumbent agents and let $A$ be a finite
nonempty set of indivisible goods, with one unit of each good. Each incumbent
$i\in N$ has a valuation
\[
v_i:2^A\to\R_+
\]
that is normalized, $v_i(\emptybundle)=0$, and monotone:
\[
v_i(B)\leq v_i(C)
\qquad
\text{whenever }B\subseteq C.
\]
Preferences are quasilinear. At a price vector $p\in\R^A$, the utility of
agent $i$ from bundle $B$ is
\[
u_i(B;p)
=
v_i(B)-p(B),
\qquad
p(B):=\sum_{a\in B}p_a,
\]
and the demand correspondence is
\[
D_{v_i}(p)
:=
\arg\max_{B\subseteq A}
\{v_i(B)-p(B)\}.
\]

The incumbent economy is denoted by
\[
\M
=
\langle N,A;(v_i)_{i\in N}\rangle.
\]

An \emph{allocation} is a collection
$\mathbf X=(X_i)_{i\in N}$ that partitions $A$: the bundles are pairwise
disjoint and
\[
\bigcup_{i\in N}X_i=A.
\]
A \emph{competitive equilibrium} is a pair $(p,\mathbf X)$ such that
$\mathbf X$ is an allocation and
\[
X_i\in D_{v_i}(p)
\qquad
\text{for every }i\in N.
\]
Empty bundles are allowed.

For every $B\subseteq A$, define the \emph{aggregate welfare valuation} of
the incumbent population by
\begin{equation}
\label{eq:social-valuation}
v_N(B)
:=
\max\left\{
\sum_{i\in N}v_i(B_i):
(B_i)_{i\in N}\text{ is a partition of }B
\right\}.
\end{equation}
Because the economy is finite, the maximum is attained. An allocation
$\mathbf X$ is \emph{efficient} if
\[
\sum_{i\in N}v_i(X_i)=v_N(A).
\]

A valuation $v$ is \emph{additive} if there exists a vector $q\in\R_+^A$
such that
\[
v(B)=q(B)=\sum_{a\in B}q_a
\qquad
\text{for every }B\subseteq A.
\]

Let $w:2^A\to\R_+$ be a normalized monotone valuation of a new agent,
denoted by $*$. The corresponding one-agent extension is
\[
\M\oplus w
:=
\langle
N\cup\{*\},A;(v_i)_{i\in N},w
\rangle.
\]

\begin{definition}[Universal entry robustness]
\label{def:robustness}
The incumbent economy $\M$ is \emph{universally entry-robust} if $\M$
admits a competitive equilibrium and, for every normalized monotone entrant
valuation $w$, the augmented economy $\M\oplus w$ admits a competitive
equilibrium.
\end{definition}

Universal entry robustness is an exact population-inclusion requirement. The
entrant is not assumed to be small relative to the incumbent economy, and no
restriction is imposed on the entrant's complementarities.

\begin{example}[Entry may destroy welfare decentralization]
\label{ex:fragility}
Let $A=\{a,b\}$ and suppose there is one incumbent with valuation
\[
v(\{a\})=v(\{b\})=0,
\qquad
v(\{a,b\})=1.
\]
The incumbent economy has a competitive equilibrium at zero prices, with
both goods assigned to the incumbent.

Now add a unit-demand entrant with
\[
w(\{a\})=w(\{b\})=w(\{a,b\})=\frac34.
\]
The unique efficient allocation assigns both goods to the incumbent. If this
allocation were supported by prices $p$, the entrant's demand for the empty
bundle would require
\[
p_a\geq\frac34
\qquad\text{and}\qquad
p_b\geq\frac34.
\]
The incumbent's demand for $\{a,b\}$ over the empty bundle would require
\[
p_a+p_b\leq1,
\]
a contradiction. Hence entry destroys the possibility of decentralizing the
efficient allocation through item prices.
\end{example}

\section{Two preliminary welfare observations}
\label{sec:preliminaries}

We first record two facts that allow us to move between the individual economy
and its aggregate welfare representation.

\begin{lemma}[Common support of efficient allocations]
\label{lem:common-support}
Suppose $(p,\mathbf X)$ is a competitive equilibrium of a finite economy.
Then $\mathbf X$ is efficient. Moreover, if $\mathbf Y$ is any other
efficient allocation, then $(p,\mathbf Y)$ is also a competitive equilibrium.
\end{lemma}

\begin{proof}
Let $\mathbf Y=(Y_i)_{i\in N}$ be any allocation. Since
$X_i\in D_{v_i}(p)$,
\[
v_i(X_i)-p(X_i)
\geq
v_i(Y_i)-p(Y_i)
\qquad
\text{for every }i\in N.
\]
Summing over agents and using
\[
\sum_{i\in N}p(X_i)
=
p(A)
=
\sum_{i\in N}p(Y_i)
\]
gives
\[
\sum_{i\in N}v_i(X_i)
\geq
\sum_{i\in N}v_i(Y_i).
\]
Thus $\mathbf X$ is efficient.

If $\mathbf Y$ is also efficient, the last inequality is an equality. The
individual demand inequalities are all weak and their sum is zero, so each is
an equality. Hence
\[
Y_i\in D_{v_i}(p)
\qquad
\text{for every }i\in N,
\]
and $(p,\mathbf Y)$ is a competitive equilibrium.
\end{proof}

The second lemma shows that a competitive equilibrium of an augmented economy
induces one in which the incumbent population is represented by $v_N$.

\begin{lemma}[Aggregation]
\label{lem:aggregation}
Let $w$ be an entrant valuation. If $\M\oplus w$ admits a competitive
equilibrium at prices $p$, then the two-agent economy with valuations $v_N$
and $w$ also admits a competitive equilibrium at prices $p$.
\end{lemma}

\begin{proof}
Let
\[
\bigl(p,(X_i)_{i\in N},Y\bigr)
\]
be a competitive equilibrium of $\M\oplus w$, and set
\[
X_N
:=
\bigcup_{i\in N}X_i
=
A\setminus Y.
\]
We show that $X_N\in D_{v_N}(p)$.

Fix any $B\subseteq A$, and choose a partition $(B_i)_{i\in N}$ of $B$
attaining $v_N(B)$. By the definition of $v_N$ and the incumbent demand
inequalities,
\begin{align*}
v_N(X_N)-p(X_N)
&\geq
\sum_{i\in N}\bigl[v_i(X_i)-p(X_i)\bigr]
\\
&\geq
\sum_{i\in N}\bigl[v_i(B_i)-p(B_i)\bigr]
\\
&=
v_N(B)-p(B).
\end{align*}
Thus $X_N\in D_{v_N}(p)$. Since $Y\in D_w(p)$, the allocation $(X_N,Y)$
and prices $p$ form a competitive equilibrium of the representative two-agent
economy.
\end{proof}

\section{The canonical welfare-loss entrant}
\label{sec:canonical}

For a normalized monotone valuation $v$, define its \emph{complementary dual}
by
\begin{equation}
\label{eq:dual}
v^{\dagger}(B)
:=
v(A)-v(A\setminus B)
\qquad
\text{for every }B\subseteq A.
\end{equation}
The valuation $v^{\dagger}$ is normalized and monotone. It evaluates a bundle
by the loss in value caused by removing that bundle from the full resource
set.

The next proposition shows that one complementary-dual entrant is sufficient
to test whether an aggregate welfare valuation can be decentralized robustly.

\begin{proposition}[Canonical welfare-loss test]
\label{prop:stress-test}
Let $v$ be a normalized monotone valuation. The two-agent economy with
valuations $v$ and $v^{\dagger}$ admits a competitive equilibrium if and only
if $v$ is additive.
\end{proposition}

\begin{proof}
Suppose first that the two-agent economy admits a competitive equilibrium at
prices $p$. For every $S\subseteq A$, the allocation assigning $S$ to the
first agent and $A\setminus S$ to the second has welfare
\[
v(S)+v^{\dagger}(A\setminus S)
=
v(S)+v(A)-v(S)
=
v(A).
\]
Thus every allocation is efficient. By Lemma~\ref{lem:common-support}, the
same prices $p$ support the allocation $(A,\emptybundle)$.

Fix $B\subseteq A$. Since the second agent demands the empty bundle,
\[
0
\geq
v^{\dagger}(B)-p(B),
\]
and hence
\begin{equation}
\label{eq:lower-price}
p(B)
\geq
v(A)-v(A\setminus B).
\end{equation}
Since the first agent demands $A$ rather than $A\setminus B$,
\[
v(A)-p(A)
\geq
v(A\setminus B)-p(A\setminus B),
\]
which implies
\begin{equation}
\label{eq:upper-price}
p(B)
\leq
v(A)-v(A\setminus B).
\end{equation}
Combining \eqref{eq:lower-price} and \eqref{eq:upper-price},
\begin{equation}
\label{eq:price-dual}
p(B)
=
v(A)-v(A\setminus B)
\qquad
\text{for every }B\subseteq A.
\end{equation}

Taking $B=A$ gives $p(A)=v(A)$. Therefore, for every $S\subseteq A$,
\[
v(S)
=
v(A)-p(A\setminus S)
=
p(A)-p(A\setminus S)
=
p(S).
\]
Since item prices are additive across goods, $v$ is additive.

Conversely, suppose that $v$ is additive and write
\[
v(B)=q(B)
\qquad
\text{for every }B\subseteq A.
\]
Then
\[
v^{\dagger}(B)
=
q(A)-q(A\setminus B)
=
q(B).
\]
At prices $q$, every bundle gives utility zero to both agents. Hence every
partition of $A$ is a competitive equilibrium allocation.
\end{proof}

For the incumbent economy $\M$, define the \emph{canonical welfare-loss
entrant} by
\begin{equation}
\label{eq:canonical-entrant}
w^{\M}(B)
:=
v_N^{\dagger}(B)
=
v_N(A)-v_N(A\setminus B).
\end{equation}
For every bundle $B$, $w^{\M}(B)$ is exactly the maximal welfare that the
incumbent population forgoes when the entrant receives $B$.

\section{Universal entry robustness}
\label{sec:characterization}

We now characterize the incumbent economies whose efficient allocations can
be decentralized after every possible one-agent extension.

\begin{theorem}[Universal entry robustness]
\label{thm:main}
Let
\[
\M
=
\langle N,A;(v_i)_{i\in N}\rangle
\]
be a finite economy with normalized monotone valuations. The following
statements are equivalent:
\begin{enumerate}
\item[(i)]
$\M$ is universally entry-robust.

\item[(ii)]
The aggregate welfare valuation $v_N$ is additive.

\item[(iii)]
The augmented economy $\M\oplus w^{\M}$ admits a competitive equilibrium.

\item[(iv)]
There exists a price vector $q\in\R_+^A$ such that $\M$ admits a competitive
equilibrium at $q$ and, for every normalized monotone entrant valuation $w$,
the augmented economy $\M\oplus w$ admits a competitive equilibrium at the
same prices $q$.
\end{enumerate}
When these conditions hold, the common price vector in \textup{(iv)} can be
chosen as
\[
q_a=v_N(\{a\})
\qquad
\text{for every }a\in A.
\]
\end{theorem}

\begin{proof}
We prove
\[
\text{(ii)}
\Longrightarrow
\text{(iv)}
\Longrightarrow
\text{(i)}
\Longrightarrow
\text{(iii)}
\Longrightarrow
\text{(ii)}.
\]

Suppose first that \textup{(ii)} holds. Let $q\in\R_+^A$ satisfy
\[
v_N(B)=q(B)
\qquad
\text{for every }B\subseteq A.
\]
Necessarily,
\[
q_a=v_N(\{a\})
\qquad
\text{for every }a\in A.
\]

For each incumbent $i$ and each bundle $B$,
\begin{equation}
\label{eq:individual-upper-bound}
v_i(B)
\leq
v_N(B)
=
q(B),
\end{equation}
because assigning $B$ to agent $i$ and the empty bundle to every other
incumbent is feasible in the definition of $v_N(B)$.

Choose a partition $(X_i)_{i\in N}$ of $A$ attaining $v_N(A)$. By
\eqref{eq:individual-upper-bound}, each difference
$q(X_i)-v_i(X_i)$ is nonnegative, while
\[
\sum_{i\in N}\bigl[q(X_i)-v_i(X_i)\bigr]
=
q(A)-v_N(A)
=
0.
\]
Hence
\[
v_i(X_i)=q(X_i)
\qquad
\text{for every }i\in N.
\]
Together with \eqref{eq:individual-upper-bound}, this implies
\[
v_i(X_i)-q(X_i)
=
0
\geq
v_i(B)-q(B)
\qquad
\text{for every }B\subseteq A.
\]
Thus
\[
\bigl(q,(X_i)_{i\in N}\bigr)
\]
is a competitive equilibrium of $\M$.

Now fix an arbitrary normalized monotone entrant valuation $w$. Choose
\[
Y\in D_w(q)
\]
and let
\[
R:=A\setminus Y.
\]
Choose a partition $(Y_i)_{i\in N}$ of $R$ attaining $v_N(R)$. As before,
\[
\sum_{i\in N}\bigl[q(Y_i)-v_i(Y_i)\bigr]
=
q(R)-v_N(R)
=
0,
\]
and every term is nonnegative by \eqref{eq:individual-upper-bound}.
Therefore,
\[
v_i(Y_i)=q(Y_i)
\quad\text{and}\quad
Y_i\in D_{v_i}(q)
\qquad
\text{for every }i\in N.
\]
Since $Y\in D_w(q)$, the allocation
\[
\bigl((Y_i)_{i\in N},Y\bigr)
\]
and prices $q$ form a competitive equilibrium of $\M\oplus w$. This proves
\textup{(iv)}.

The implication
\[
\text{(iv)}\Longrightarrow\text{(i)}
\]
follows directly from Definition~\ref{def:robustness}, and
\[
\text{(i)}\Longrightarrow\text{(iii)}
\]
follows by choosing the entrant valuation $w=w^{\M}$.

Finally, suppose \textup{(iii)} holds. By Lemma~\ref{lem:aggregation}, the
two-agent economy with valuations $v_N$ and
\[
w^{\M}=v_N^{\dagger}
\]
admits a competitive equilibrium. Proposition~\ref{prop:stress-test} then
implies that $v_N$ is additive. Hence
\[
\text{(iii)}\Longrightarrow\text{(ii)}.
\]
\end{proof}

Theorem~\ref{thm:main} yields a one-entrant test for a requirement quantified
over all entrants.

\begin{corollary}[One-entrant test]
\label{cor:one-test}
The economy $\M$ is universally entry-robust if and only if
$\M\oplus w^{\M}$ admits a competitive equilibrium.
\end{corollary}

The common price system is uniquely determined by the aggregate welfare
valuation and can itself be identified from the canonical extension.

\begin{corollary}[Uniform and canonical prices]
\label{cor:uniform-price}
Suppose that $\M$ is universally entry-robust. Then
\[
q_a=v_N(\{a\})
\qquad
\text{for every }a\in A
\]
is the unique price vector with the following property: $\M$ admits a
competitive equilibrium at $q$, and, for every normalized monotone entrant
valuation $w$, the augmented economy $\M\oplus w$ admits a competitive
equilibrium at the same prices $q$.

Moreover, every competitive equilibrium of the canonical extension
$\M\oplus w^{\M}$ has price vector $q$.
\end{corollary}

\begin{proof}
Existence follows from Theorem~\ref{thm:main}. Let $p$ be any price vector
that supports the incumbent economy and every one-agent extension. In
particular, $p$ supports the canonical extension $\M\oplus w^{\M}$. By
Lemma~\ref{lem:aggregation}, $p$ supports a competitive equilibrium in the
two-agent economy $(v_N,v_N^{\dagger})$. Applying
\eqref{eq:price-dual} to $v_N$ gives
\[
p(B)=v_N(B)
\qquad
\text{for every }B\subseteq A.
\]
Hence $p_a=v_N(\{a\})=q_a$ for every $a\in A$, proving uniqueness. The same
argument applies to the price vector of any competitive equilibrium of the
canonical extension.
\end{proof}

\section{Welfare interpretation}
\label{sec:welfare}

This section records three implications of the characterization for welfare
aggregation and decentralization.

\subsection{A social rather than individual restriction}

Universal entry robustness imposes no common domain restriction on the
individual incumbent valuations. The condition is instead imposed on their
aggregate welfare opportunities. Individual complementarities can offset one
another through reassignment across incumbents, leaving a linear aggregate
valuation.

\begin{example}[Nonadditive individuals and additive aggregate welfare]
\label{ex:aggregate-additivity}
Let $A=\{a,b\}$ and $N=\{1,2\}$. Consider
\[
\begin{array}{c|ccc}
 & \{a\} & \{b\} & \{a,b\} \\
\hline
v_1 & 1 & 1 & 1 \\
v_2 & 0 & 0 & 2
\end{array}
\]
with both valuations equal to zero at the empty bundle. Neither valuation is
additive. Nevertheless,
\[
v_N(\{a\})=1,
\qquad
v_N(\{b\})=1,
\qquad
v_N(\{a,b\})=2.
\]
Thus $v_N(B)=q(B)$ at $q=(1,1)$. By
Theorem~\ref{thm:main}, the economy is universally entry-robust, and the same
prices $(1,1)$ support a competitive equilibrium after every possible
single-agent entry.
\end{example}

The example shows that universal entry robustness is not a disguised form of
individual additivity. It is generated jointly by heterogeneity and the
possibility of reallocating goods within the incumbent population.

\subsection{Entry as a welfare decomposition problem}

Suppose that $v_N=q$ is additive. For an entrant valuation $w$, the greatest
welfare attainable after entry is
\begin{align*}
\max_{Y\subseteq A}
\bigl\{w(Y)+v_N(A\setminus Y)\bigr\}
&=
\max_{Y\subseteq A}
\bigl\{w(Y)+q(A\setminus Y)\bigr\}
\\
&=
q(A)
+
\max_{Y\subseteq A}
\bigl\{w(Y)-q(Y)\bigr\}.
\end{align*}
Therefore, every entrant demand $Y\in D_w(q)$ identifies an efficient division
of resources between the entrant and the incumbent population. Any incumbent
allocation attaining $v_N(A\setminus Y)$ completes this division into an
efficient allocation among all agents.

\begin{corollary}[Robust welfare decomposition]
\label{cor:welfare-decomposition}
Suppose $v_N$ is additive and represented by $q$. For every entrant valuation
$w$ and every $Y\in D_w(q)$,
\[
w(Y)+v_N(A\setminus Y)
=
\max_{B\subseteq A}
\bigl\{w(B)+v_N(A\setminus B)\bigr\}.
\]
Moreover, every incumbent partition of $A\setminus Y$ attaining
$v_N(A\setminus Y)$ can be combined with $Y$ to form a competitive and
efficient allocation at prices $q$.
\end{corollary}

Thus the common price vector does two jobs. It decentralizes the internal
allocation problem of the incumbent population and determines the
welfare-maximizing boundary between the incumbent population and every
possible entrant.

\subsection{The canonical entrant as an opportunity-cost test}

The canonical entrant $w^{\M}$ values each bundle at the welfare loss imposed
on the incumbents. Consequently,
\[
v_N(A\setminus B)+w^{\M}(B)=v_N(A)
\qquad
\text{for every }B\subseteq A.
\]
Every division between the representative incumbent population and the
canonical entrant is therefore socially indifferent. The test asks whether a
single anonymous item-price vector can decentralize this entire set of
welfare-equivalent divisions. The answer is affirmative exactly when the
aggregate welfare valuation itself is generated by itemwise social
opportunity costs.

This also clarifies the strength of universal entry robustness. Ordinary
equilibrium existence requires price support for at least one efficient
allocation in the realized economy. Universal entry robustness requires the
incumbent welfare technology to remain price-decentralizable after every
possible change in the boundary between incumbent and entrant consumption.
Aggregate additivity is precisely the condition that makes this possible.

The mechanism can be summarized as follows. The canonical entrant makes every
boundary between incumbent and entrant consumption welfare-equivalent.
Universal price decentralization is therefore possible only if the
incumbents' welfare loss from transferring a bundle is itself representable
by anonymous item prices. This is exactly aggregate additivity.

\section{Relation to existing equilibrium-existence characterizations}
\label{sec:existing}

The proof of Theorem~\ref{thm:main} is intentionally direct. It uses only
individual demand inequalities, efficiency, and the welfare interpretation of
the canonical entrant. This section explains how the same characterization
is represented in broader equilibrium-existence frameworks and clarifies what
the direct argument adds.

\subsection{A robust-integrality interpretation}

\citet{BikhchandaniMamer1997} show that a fixed economy with indivisible goods
admits a competitive equilibrium if and only if the linear relaxation of its
configuration problem has an integral optimal solution. Combining their
characterization with Proposition~\ref{prop:stress-test} and
Theorem~\ref{thm:main} yields an equivalent formulation of our result.

\begin{corollary}[Robust integrality]
\label{cor:robust-integrality}
Let $v_N$ be the incumbents' aggregate welfare valuation. The following
statements are equivalent:
\begin{enumerate}
\item[(i)]
$v_N$ is additive.

\item[(ii)]
For every normalized monotone valuation $w$, the configuration problem of the
representative two-agent economy $(v_N,w)$ has an integral optimal solution.

\item[(iii)]
The configuration problem of $(v_N,v_N^{\dagger})$ has an integral optimal
solution.
\end{enumerate}
\end{corollary}

\begin{proof}
By the characterization of \citet{BikhchandaniMamer1997}, statements
\textup{(ii)} and \textup{(iii)} are respectively equivalent to competitive-
equilibrium existence in every representative economy $(v_N,w)$ and in the
canonical representative economy $(v_N,v_N^{\dagger})$. The equivalence then
follows from Proposition~\ref{prop:stress-test} and
Theorem~\ref{thm:main}.
\end{proof}

Corollary~\ref{cor:robust-integrality} shows how the present characterization
appears in the linear-programming framework. Aggregate additivity is exactly
the condition under which appending an arbitrary valuation preserves
integrality, and the canonical welfare-loss valuation is sufficient to test
this robust-integrality property.

The two approaches answer different questions. The linear-programming
characterization determines equilibrium existence for arbitrary fixed
valuation profiles. The direct argument identifies the primitive welfare
structure responsible for robustness to arbitrary entry. In the canonical
extension, every division of the goods produces the same welfare. Price
support therefore forces anonymous item prices to reproduce the entire
aggregate welfare valuation, which is possible exactly under additivity. The
direct argument also identifies the endogenous stress test $v_N^{\dagger}$
and the unique common price vector that supports every extension; these
objects arise from the economics of the robustness question rather than from
the solution of a primal--dual program.

\subsection{Profile-specific robustness and domain-wide existence}

The quantifiers in Theorem~\ref{thm:main} differ from those in domain-wide
existence results. \citet{BaldwinKlemperer2019} classify valuations through
demand types and characterize broad domains for equilibrium existence using
unimodularity. \citet{JagadeesanTeytelboym2025} formulate economic conditions
on substitutability and complementarity across bundles that are sufficient
and essentially necessary for equilibrium existence in rich domains. These
approaches seek restrictions under which all relevant profiles assembled from
a valuation domain admit equilibrium.

Universal entry robustness is asymmetric and profile-specific. The incumbent
economy, summarized by $v_N$, is fixed, while the valuation of one additional
agent is unrestricted. The family
\[
\{(v_N,w): w\text{ is normalized and monotone}\}
\]
need not itself form a unimodular or bundle-consistent domain. Nevertheless,
when $v_N$ is additive, every profile in this family admits equilibrium at the
same prices. The reason is not a common restriction on the entrant's demand,
but the flexibility of the additive aggregate incumbent valuation at its
supporting prices: the entrant may choose any demanded bundle, and the
incumbent population can absorb the complement.

Accordingly, the present characterization complements domain-wide results.
Demand-type and bundle-consistency approaches explain which patterns of
preferences can coexist throughout a domain. Theorem~\ref{thm:main} instead
identifies when one realized economy can accommodate an arbitrary additional
preference. The distinction is between robustness of a valuation domain and
robustness of an incumbent welfare structure.

\section{Conclusion}
\label{sec:conclusion}

We have formulated a population-robustness question for competitive
allocation with indivisible goods: when can a fixed incumbent economy
accommodate every possible single entrant while retaining efficient
decentralization through anonymous item prices? The answer is entirely
determined by aggregate welfare. Universal entry robustness holds if and only
if the incumbents' aggregate welfare valuation is additive.

The characterization yields two stronger conclusions. First, a requirement
quantified over all entrant valuations can be tested using one endogenous
welfare-loss entrant. Second, whenever robustness holds, the unique price
vector
\[
q_a=v_N(\{a\})
\]
supports the incumbent economy and every one-agent extension. The same prices
therefore decentralize both the incumbents' internal allocation problem and
the welfare-maximizing boundary between the incumbents and any entrant.

The result can also be recovered from the linear-programming characterization
of \citet{BikhchandaniMamer1997}, where it appears as robust integrality of the
relevant configuration problems. The direct welfare argument serves a
different purpose. By making every entrant--incumbent division
welfare-equivalent, the canonical entrant shows immediately why price support
forces aggregate welfare to be additive. It thereby separates the economic
mechanism from the general mathematical machinery used to test equilibrium
existence.

The universal requirement is intentionally strong. Natural extensions are to
restrict the class of entrants, to study robustness to several entrants, and
to compare exact entry robustness with approximate robustness in large
replica economies. These questions may reveal intermediate aggregate welfare
structures between unrestricted nonadditivity and the full additivity
identified here.


\section*{Statements and Declarations}

\noindent\textbf{Competing interests.}
The author declares no competing interests.

\medskip

\noindent\textbf{Data availability.}
No datasets were generated or analyzed in this study.

\end{document}